\documentclass[12pt]{article}%
\pdfoutput=1
\usepackage{graphicx}
\usepackage{amsmath}
\usepackage{amsfonts}
\usepackage{amssymb}%
\providecommand{\U}[1]{\protect\rule{.1in}{.1in}}
\newtheorem{theorem}{Theorem}
\newtheorem{acknowledgement}[theorem]{Acknowledgement}

\newtheorem{lemma}[theorem]{Lemma}

\newenvironment{proof}[1][Proof]{\noindent\textbf{#1.} }{\ \rule{0.5em}{0.5em}}
\begin{document}

\title{\textbf{A manifold of pure Gibbs states of the Ising model on a Cayley tree}}
\author{Daniel Gandolfo, Jean Ruiz\\Centre de Physique Th\'{e}orique, Aix-Marseille Univ, \\CNRS UMR 7332, Univ Sud Toulon Var, \\13288 Marseille cedex 9, France;
\and Senya Shlosman\\Centre de Physique Th\'{e}orique, Aix-Marseille Univ, \\CNRS UMR 7332, Univ Sud Toulon Var, \\13288 Marseille cedex 9, France;\\Inst. of the Information Transmission Problems,\\RAS, Moscow, Russia}
\maketitle

\begin{abstract}
We study the Ising model on a Cayley tree. A wide class of new Gibbs states is exhibited.

\end{abstract}

\section{Introduction}

In this paper we consider the well-studied n.n. Ising model on the Bethe
lattice $\mathcal{T}^{k}$ (or Cayley tree). The natural title of the present
paper would be "New Gibbs states for the Ising model on the Bethe lattice".
However the subject of "New States on the Bethe lattice" is not new for
already a long time, and was discussed by many authors. So a better name might
be "All pure low-temperature states" (especially if it describes correctly
what follows).

In spite of its simplicity, some questions in the area of statistical
mechanics on trees are still open, and new phenomena are encountered.

The $\left(  +\right)  $ and $\left(  -\right)  $ phases of the Ising model on
the Bethe lattice -- the translation invariant states $\mu^{+}$ and $\mu^{-}$
-- are known to be extremal below the critical temperature $T_{c}$. An
uncountable number of non-translation invariant pure states $\mu^{\pm}$ was
constructed by Blekher and Ganikhodzhaev in \cite{BG}. These states describe
the coexistence of the $\mu^{+}$ and $\mu^{-}$ states along different `rigid
interfaces' $S$, analogous to the Dobrushin states in the 3D Ising model. One
should think about $S$ as a path, connecting two points `at infinity' of
$\mathcal{T}^{k}.$ In what follows we will call these states as BG-states and
denote them also by $\mu_{S}^{\pm}$.

The \textquotedblleft free state\textquotedblright\ $\mu^{0}$, i.e. the one
corresponding to `zero' or `empty' boundary condition is extremal in the
interval below the critical temperature $T_{c}$ and above the critical
temperature of the corresponding spin-glass model (see \cite{BRZ}). In
particular, it is not extremal at low temperatures. One does not know the
decomposition of $\mu^{0}$ into extremal states (a question of A. van Enter).
We will formulate below our conjecture about this decomposition.

Later Rakhmatullaev and Rozikov \cite{RR} introduced some new states, which
they call weakly periodic states. Their construction is somewhat indirect and
uses various subgroups of the free group.

In what follows we present an alternative construction of weakly periodic
states, as well as the construction of many more pure states. This
construction is a result of our attempt to understand the claim of \cite{RR}.

\section{New ground state configurations}

Let $\mathcal{T}^{k}=\left(  V,E\right)  $ be the uniform Cayley tree, where
each vertex has $k+1$ neighbors. Let $0\in V$ be its root. We consider the
n.n. ferromagnetic Ising model on $\mathcal{T}^{k}$ with interaction $J>0.$

The following simple remark, valid for the Ising model on any rooted tree,
will be very helpful.

Let $D\subset E$ be an arbitrary collection (finite or infinite) of edges of
$\mathcal{T}^{k}.$ Define the configurations $\sigma^{D+},\sigma^{D-}$ on
$\mathcal{T}^{k},$ which satisfy the property: for every bond $b\in E,$
$b=\left(  b_{1},b_{2}\right)  \subset V$%
\[
\sigma_{b_{1}}^{D+}\sigma_{b_{2}}^{D+}=\sigma_{b_{1}}^{D-}\sigma_{b_{2}}%
^{D-}=\left\{
\begin{array}
[c]{cc}%
-1 & \text{ for }b\in D,\\
+1 & \text{ for }b\notin D,
\end{array}
\right.
\]%
\[
\sigma_{0}^{D+}=+1,\ \sigma_{0}^{D-}=-1.
\]
The existence and uniqueness of the configurations $\sigma^{D+}$ and
$\sigma^{D-}$ is immediate. Note that the map $D\rightsquigarrow\sigma^{D+}$
is a one-to-one map between the subsets $D\subseteq E$ and the spin
configurations on $\mathcal{T}^{k}$ which have the value $+1$ at the root
$0\in V.$

Suppose a set $D\subseteq E$ is given. For every vertex $v\in V$ define
$d_{D}\left(  v\right)  $ to be the number of bonds in $D,$ which are incident
to $v.$ Define the number
\[
d_{D}=\max_{v}d_{D}\left(  v\right)  .
\]
Suppose now that the set $D$ is such that
\begin{equation}
d_{D}<(k-1)/2. \label{01}%
\end{equation}
(For example, if all bonds in $D$ are disjoint, then $d_{D}=1,$ in which case
$k$ has to be at least $4$). Then the following theorem holds.

\begin{theorem}
The configurations $\sigma^{D+}$ and $\sigma^{D-}$ are ground state configurations.
\end{theorem}

\begin{center}

\noindent
\end{center}

Let $\sigma$ be a ground state configuration, and $C\subset V$ be a connected
finite set. The configuration $\sigma_{C}$, obtained from $\sigma$ by flipping
all the spins of $\sigma$ in $C,$ will be called a \textit{connected
excitation} of $\sigma.$ We will call the ground state configuration $\sigma$
\textit{stable }(see \cite{DS}), if for any $v\in V,$ any excess energy level
$E>0,$ the number of connected excitations $\sigma_{C}$ of $\sigma,$
satisfying%
\[
H\left(  \sigma_{C}\right)  -H\left(  \sigma\right)  <E,\ \ v\in C
\]
is bounded uniformly in $v.$

\begin{theorem}
Under the condition $\left(  \ref{01}\right)  $ the ground state
configurations $\sigma^{D+}$ and $\sigma^{D-}$ are stable.
\end{theorem}

\section{New low-T Gibbs states}

Denote by $\mu^{D_{+}}$ and $\mu^{D_{-}}$ the Gibbs states corresponding to
the boundary conditions $\sigma^{D+}$ and $\sigma^{D-}.$ The quantification of
the last statement allows us to show the following

\begin{theorem}
\label{LowT} There exists a value $T_{k}$ of the temperature, such that for
all temperatures $T<T_{k}$ and all collections $D$ satisfying $\left(
\ref{01}\right)  $ the states $\mu^{D_{+}}$and $\mu^{D_{-}}$ are extremal
Gibbs states. (Here the temperature $T_{k}$ does not depend on $D$).
\end{theorem}

These low-temperature Gibbs states are small perturbations of the ground
states $\sigma^{D+}$ and $\sigma^{D-}$. It is interesting to observe that the
free energy of the states $\mu^{D+}$ and $\mu^{D-}$ is higher than that of the
states $\mu^{+}$, $\mu^{-}\ $and $\mu^{\pm}$.

\section{Examples and comments}
\begin{enumerate}

\item $D=\varnothing.$ In that case the states $\mu^{\varnothing+}$ and
$\mu^{\varnothing-}$ are just the $\left(  +\right)  $ and $\left(  -\right)
$ states.

\item In the case when $D$ consists from a single bond $b,$ the states $\mu^{b+}$
and $\mu^{b-}$ are among the non-translation invariant states $\mu_{S}^{\pm}$
constructed by Blekher and Ganikhodzhaev in \cite{BG}. The interface curve $S$
in that case is just the curve intersecting $\mathcal{T}^{k}$ along the single
bond $b.$ The same is true for every finite collection $D.$ Other states
constructed in \cite{BG} correspond to some countable families $D.$

\item Consider the case when $D=\bar{D}\ $-- a dimer covering of $\mathcal{T}%
^{k}.$ That means that for every vertex $v\in V$ there exists exactly one bond
$d\left(  v\right)  \in\bar{D},$ such that $v\in d\left(  v\right)  .$ The
existence of dimer coverings is straightforward.

\begin{center}
\includegraphics[width=10cm]{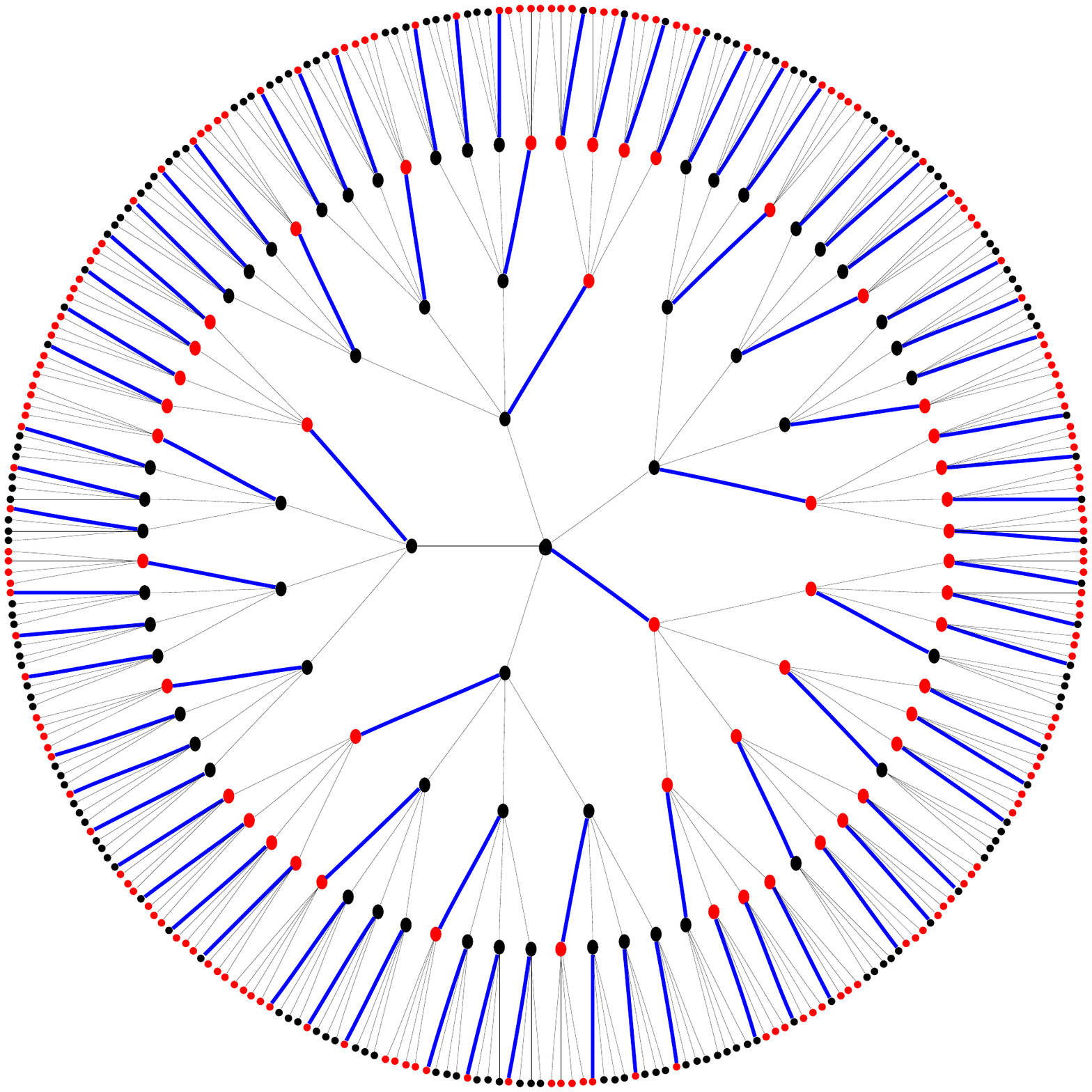}

{\footnotesize {Fig. 1. The configuration $\sigma^{\bar{D}+}$ on the Cayley
tree $\mathcal{T}^{4}$ (for four generations). Black and red dots correspond
to ($-$) and ($+$) spins.}}

{\footnotesize {Here }}$\bar{D}$ {\footnotesize is a dimer covering of
{$\mathcal{T}^{4}$ (blue bonds).} }
\end{center}

In the terminology of \cite{RR} the ground state configurations $\sigma
^{\bar{D}+}$ and $\sigma^{\bar{D}-}$ are called weakly periodic. According to
our theorems they generate the corresponding low-temperature extremal states
$\mu^{\bar{D}+}$ and $\mu^{\bar{D}-}$. The qualitative difference between the
states $\mu^{\bar{D}\pm}$ and the BG states $\mu_{S}^{\pm}$ is the following.
The states $\mu_{S}^{\pm}$ describe the coexistence of the $\left(  +\right)
$ and $\left(  -\right)  $ states separated by an interface $S$; that means
for example that for every $R>0$ one can find (many) points $x\in V,$ such
that the phase in the ball $U_{R}\left(  x\right)  $ looks as a small
perturbation of the $\left(  +\right)  $ or $\left(  -\right)  $ state.
Namely, any point $x\in V$ with $\mathrm{dist}\left(  x,S\right)  >\left(
1+\varepsilon\right)  R$ would go. On the other hand, the typical
configuration of the state $\mu^{\bar{D}+}$ or $\mu^{\bar{D}-}$ around any
point $x\in V$ has certain fractions of $\left(  +\right)  $ and $\left(
-\right)  $ spins.

\item Let $f_{D}\left(  \beta\right)  $ be the free energy of the Gibbs state of
our Ising model at inverse temperature $\beta,$ corresponding to the boundary
conditions $\sigma^{D\pm}.$ Unlike the classical case of the $\mathbb{Z}^{\nu
}$ lattices, the free energy of the Ising model on $\mathcal{T}^{k}$ might
depend on the boundary conditions. Clearly, for every $\beta$%
\[
f_{\bar{D}}\left(  \beta\right)  \geq f_{D}\left(  \beta\right)  \geq
f_{\varnothing}\left(  \beta\right)  ,
\]
moreover,%
\[
f_{\bar{D}}\left(  \beta\right)  >f_{\varnothing}\left(  \beta\right)
\]
for $\beta$ large enough. We conjecture that the states $\mu^{\bar{D}+}$ and
$\mu^{\bar{D}-}$ have the highest possible free energy among all the Gibbs
states at a given temperature on the tree $\mathcal{T}^{k}$ with $k=4$.

\item A more elaborate example is the `secondary dimer covering', corresponding
to a dimer covering $\bar{D}$. By this we mean any collection of bonds
$\tilde{D},$ such that

\begin{itemize}
\item every bond $b\in\tilde{D}$ is incident to exactly two bonds $d^{\prime
},d^{\prime\prime}$ from $\bar{D}$;

\item every bond $d\in\bar{D}$ is incident to exactly one bond $b\ $from
$\tilde{D}$.
\end{itemize}

\begin{center}
\includegraphics[width=10cm]{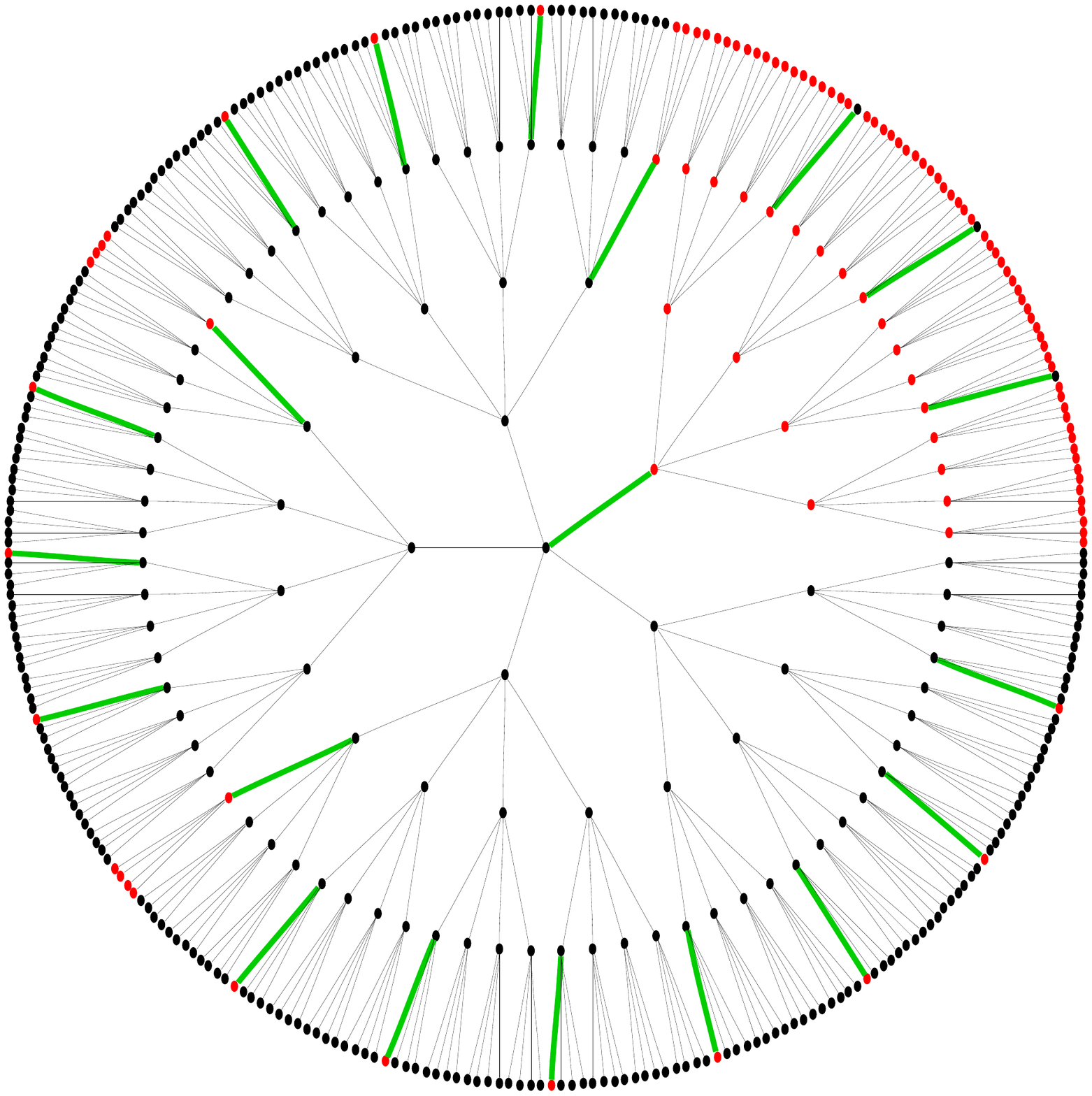}

{\footnotesize {Fig. 2. Secondary dimer covering $\tilde{D}$ (green bonds) on
the Cayley tree $\mathcal{T}^{4}$ (for four generations) corresponding to the
primary dimer covering $\bar{D}$ (blue bonds) of~Figure 1.}}
\end{center}

\noindent Again, for low temperatures we have
\[
f_{\bar{D}}\left(  \beta\right)  >f_{\tilde{D}}\left(  \beta\right)
>f_{\varnothing}\left(  \beta\right)  .
\]

\item Yet another example can be obtained by considering the family $D^{\ast},$
corresponding to the following monomer-dimer covering. By this we mean (see
Fig. 3) the collection of bonds $b\in E,$ such that

\begin{itemize}
\item for every $b^{\prime}\neq b^{\prime\prime}\in D^{\ast}$ the distance
$\mathrm{dist}\left(  b^{\prime},b^{\prime\prime}\right)  \geq2,$

\item for every $v\in V$ the distance $\mathrm{dist}\left(  v,D^{\ast}\right)
\leq1.$
\end{itemize}

\begin{center}
\includegraphics[width=10cm]{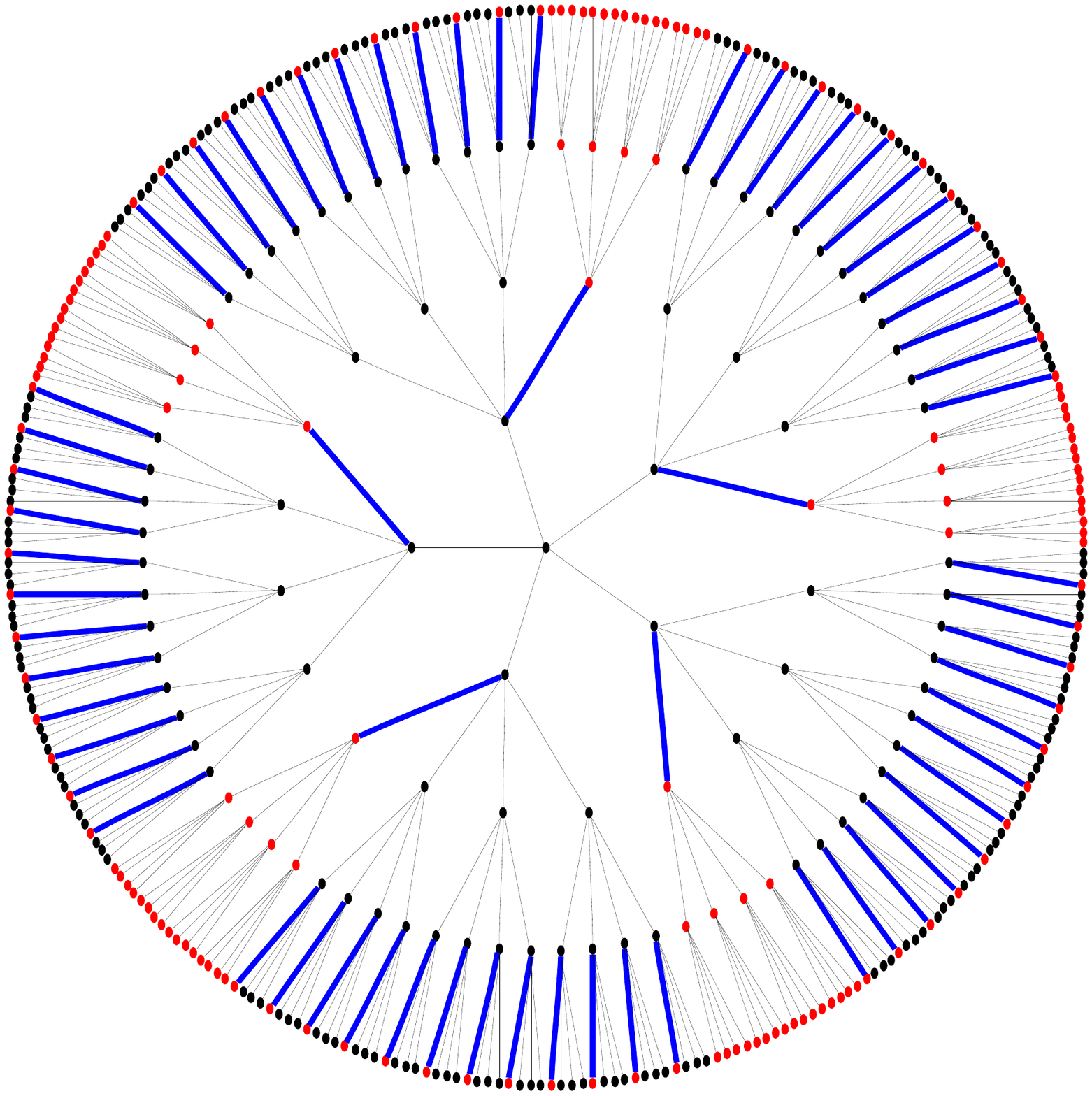}

{\footnotesize {Fig. 3. Monomer-dimer covering.}}
\end{center}

\item A final example is the family $\hat{D},$ corresponding to the `path
covering' of the tree $\mathcal{T}^{k}.$ Namely, the family $\hat{D}$ is
defined by the property that every $v\in V$ belongs to precisely two bonds
from $\hat{D}.$ Thus the bonds from $\hat{D}$ split into double infinite
non-intersecting paths, passing through each vertex of $\mathcal{T}^{k}.$ To
apply our results we need in this case $k\geq6.$

\begin{center}
\includegraphics[width=10cm]{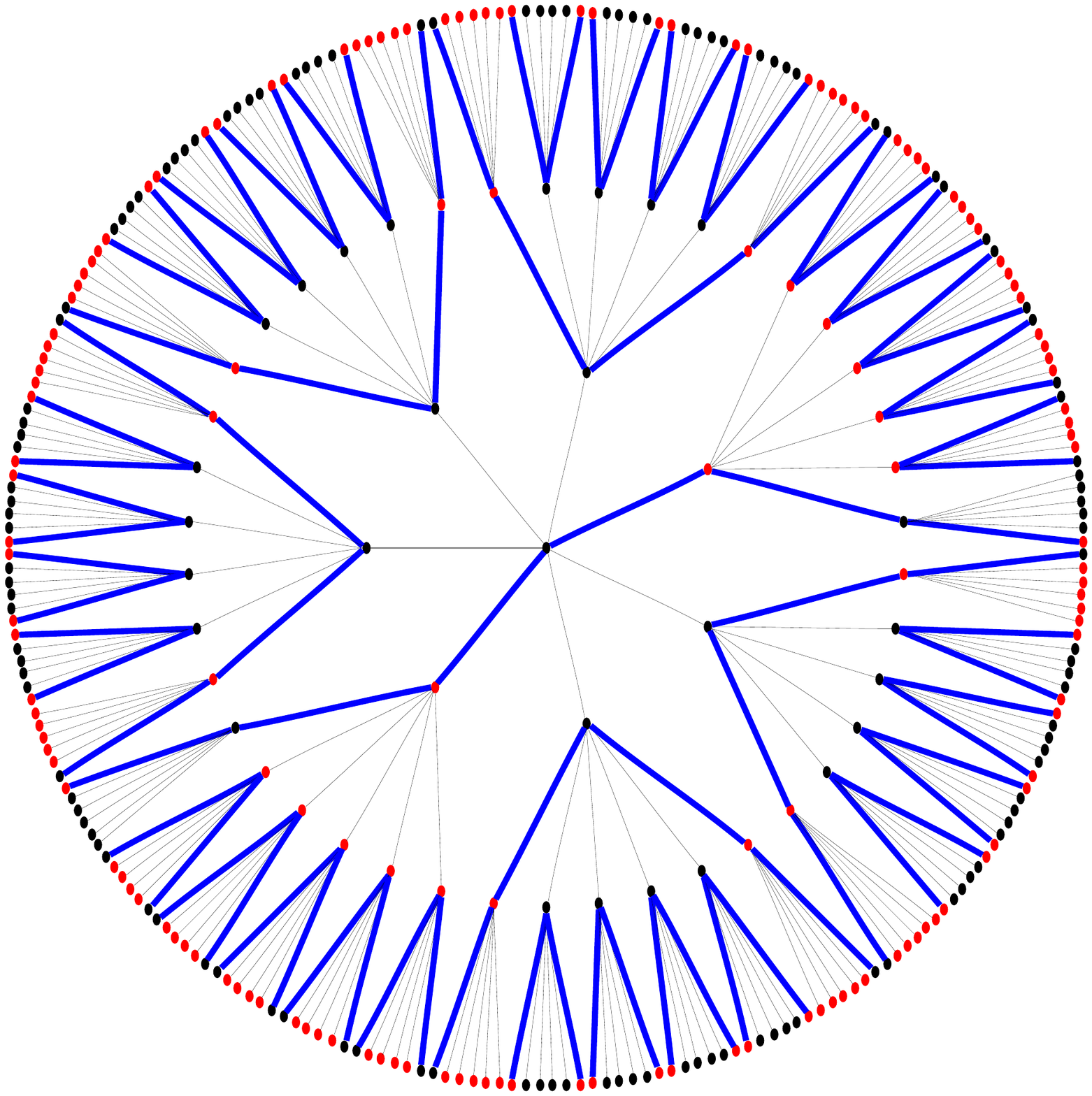}

{\footnotesize {Fig. 4. Path covering of the tree $\mathcal{T}^{6}$.}}
\end{center}

\bigskip

As the reader can easily see, many more regular, as well as infinitely more
irregular examples can be constructed in the manner above.

\item Next, we formulate our conjecture about the decomposition of the low
temperature free state $\mu^{0}$ into extremal states. We think that this
decomposition should be into our measures  $\mu^{D_{+}}$ and $\mu^{D_{-}}$
with various sets $D$ which are sparse enough (depending on the temperature).

\end{enumerate}

\section{The proofs}

The proof of Theorem \ref{LowT} is based on contour representation of the
model. The remaining theorems follow from it.

Consider a configuration $\sigma$ that coincides with a ground state
configuration $\sigma^{D}$ outside finite box $V$. A site $x\in V$ is called
incorrect, iff $\sigma_{x}\neq\sigma_{x}^{D}$. The set $J$ of incorrect sites
then decomposes into maximal connected components, $J=I_{1}\cup I_{2}\cup
\dots\cup I_{n}$. (Here a set of sites is called connected if for any two
sites $x,y$ there exist a sequence $x=x_{1},x_{2},\ldots,x_{j},\ldots x_{m}=y$
for which $x_{j}$ and $x_{j+1}$ are n.n. for all $i=1,\ldots,m-1$.)

Let us also introduce the external boundary $\partial_{E}I$ of $I$ to be the
set of vertices of the complement of $I$ that have a n.n. in $I$.

A contour $\Gamma_{i}$ of the configuration $\sigma$ is the subgraph defined
by the set of vertices $V(\Gamma_{i})=I_{i}\cup\partial_{E}I_{i}$ and all
bonds of $\mathcal{T}_{k}$ between them. The interior $\mathrm{Int}\left(
\Gamma_{i}\right)  $ of $\Gamma_{i}$ is the set $I_{i}.$

Each connected subgraph $\Gamma_{i}=\{V(\Gamma_{i}),E(\Gamma_{i})\}$
($i=1,\ldots,n$), where $V(\Gamma_{i})=I_{i}\cup\partial_{E}I_{i}$ and
$E(\Gamma_{i})$ are all the edges between n.n. pairs of points in
$V(\Gamma_{i})$, is called a contour of the configuration $\sigma$.

Obviously, there is a one-to-one correspondence between the configurations and
the sets of all its contours, provided the boundary condition $\sigma^{D}$ is
given. Every collection of contours thus constructed is called a compatible collection.

This allows to write the partition function in a finite ball $V_{r}$ of radius
$r$ as follows:
\[
Z(V_{r})=\sum_{\sigma}\exp\{-\beta\lbrack H(\sigma)-H(\sigma^{D}%
)]\}=\sum_{\{\Gamma_{1},\ldots,\Gamma_{n}\}_{c}}\prod_{i=1}^{n}\varrho
(\Gamma_{i})
\]
with $\varrho(\Gamma)=\exp\{-\beta\lbrack H(\sigma_{\Gamma})-H(\sigma^{D}%
)]\}$. Here the first sum is over all configurations that coincide with the
configuration $\sigma^{D}$ outside the ball $V_{r}$, the second sum runs over
all compatible families of contours with vertex set in $V_{r}$, and
$\sigma_{\Gamma}$ denotes the configuration that contains only one contour
$\Gamma$.

The two following lemmas provide the energy and entropy estimates for contours.

\begin{lemma}
For any contour $\Gamma$ the excess energy satisfies
\begin{equation}
H(\sigma_{\Gamma})-H(\sigma^{D})\geq2 J\left[  \left(  k+1\right)  -2\left(
d_{D}+1\right)  \right]  |\mathrm{Int}\left(  \Gamma\right)  |. \label{02}%
\end{equation}

\end{lemma}

\begin{proof}
We prove $\left(  \ref{02}\right)  $ by induction. If $\mathrm{Int}\left(
\Gamma\right)  $ consists of one point, then $H(\sigma_{\Gamma})-H(\sigma
^{D})\geq2 J\left[  \left(  k+1\right)  -2d_{D}\right]  .$ Let $\Gamma
=\{V(\Gamma),E(\Gamma)\}$, where $V(\Gamma)=I\cup\partial_{E}I,$ and suppose
that the root $0\in I.$ Let $x\in I,$ $x\neq0$ be such a site that all the $k$
offsprings of $x$ do not belong to $I.$ Note that the configuration
$\sigma_{\Gamma}^{x},$ obtained from $\sigma_{\Gamma}$ by flipping a spin at
$x,$ also has one contour, and so $\sigma_{\Gamma}^{x}=\sigma_{\Gamma^{\prime
}},$ and $\mathrm{Int}\left(  \Gamma^{\prime}\right)  =I\smallsetminus x.$ The
contribution of the site $x$ to the energy $-H(\sigma_{\Gamma}^{x})$ is at
least $J [\left(  k+1\right)  -2\left(  d_{D}+1\right) ] ;$ indeed, the site
$x$ in the configuration $\sigma_{D}$ has at most $d_{D}$ frustrated bonds,
while the number of bonds connecting $x$ to $I\smallsetminus x$ (which also
can be frustrated) is one. Since the contribution of the site $x$ to the
energy $H(\sigma_{\Gamma})$ is exactly the negative of its contribution to
$-H(\sigma_{\Gamma}^{x}),$ we have%
\[
H(\sigma_{\Gamma})-H(\sigma_{\Gamma}^{x})\geq2 J\left[  \left(  k+1\right)
-2\left(  d_{D}+1\right)  \right]  .
\]

\end{proof}

The following lemma is well-known (see e.g. \cite{O}).

\begin{lemma}
Let $G$ be a graph of maximal degree $k+1.$ Then the number of connected
subgraphs $\Gamma\subset G$ with $n$ bonds, containing a given vertex is
bounded from above by
\[
(k+1)^{2n}%
\]

\end{lemma}

With these two estimates, we have good control of cluster expansions for
partition functions and correlation functions at low temperatures. These
expansions allow to prove by known arguments that the states $\mu^{D_{+}}$and
$\mu^{D_{-}}$ are extremal.

\begin{acknowledgement}
We are grateful to A. van Enter, P. Bleher and U. Rozikov for helpful
discussions and comments.
\end{acknowledgement}


\begin{thebibliography}{999}                                                                                              %


\bibitem[BRZ]{BRZ}Bleher, P. M.; Ruiz, J.; Zagrebnov, V. A. On the purity of
the limiting Gibbs state for the Ising model on the Bethe lattice. J. Statist.
Phys. 79 (1995), no. 1-2, 473--482.

\bibitem[BG]{BG}Blekher, P. M.; Ganikhodzhaev, N. N. Pure phases of the Ising
model on Bethe lattices. (Russian) Teor. Veroyatnost. i Primenen. 35 (1990),
no. 2, 220--230; translation in Theory Probab. Appl. 35 (1990), no. 2,
216--227 (1991)

\bibitem[DS]{DS}R.L. Dobrushin and S. Shlosman: \textit{The problem of
translation-invariance of Gibbs states at low temperatures}, Mathematical
physics reviews, Vol. 5, 53--195, Soviet Sci. Rev. Sect. C: Math. Phys. Rev.,
5, Harwood Academic Publ., Chur, 1985.

\bibitem[O]{O}O. Ore, The-Four Color Problem. Academic Press, New-York,
London, 1967.

\bibitem[RR]{RR}Rozikov, U.A.; Rakhmatullaev, M.M. Description of weakly
periodic Gibbs measures for the Ising model on a Cayley tree. (English),
Theor. Math. Phys. 156, No. 2, 1218-1227 (2008); translation from Teor. Mat.
Fiz. 156, No. 2, 292-302 (2008).
\end{thebibliography}
\end{document}